\documentclass[pdflatex]{sn-jnl}

\usepackage{graphicx}%
\usepackage{multirow}%
\usepackage{amsmath,amssymb,amsfonts}%
\usepackage{amsthm}%
\usepackage{mathrsfs}%
\usepackage{booktabs}
\usepackage[title]{appendix}%
\usepackage{xcolor}%
\usepackage{textcomp}%
\usepackage{manyfoot}%
\usepackage{booktabs}%
\usepackage{algorithm}%
\usepackage{algorithmicx}%
\usepackage{algpseudocode}%
\usepackage{listings}%
\usepackage{array}

\usepackage{graphics}
\usepackage{caption}
\usepackage{enumerate}
\usepackage{enumitem}
\usepackage{array}
\usepackage{xcolor}
\usepackage{framed}
\usepackage{lineno}
\usepackage{fancyhdr,lipsum}
\usepackage{comment}
\usepackage{algorithm}
\usepackage{algpseudocode}
\usepackage{tabularx}
\usepackage{graphicx}

\theoremstyle{thmstyleone}%
\newtheorem{theorem}{Theorem}
\newcommand{\sdvs}{{\sf LaSDVS }}
\theoremstyle{thmstyletwo}%
\newtheorem{lemma}{Lemma}%
\theoremstyle{thmstylethree}%
\newtheorem{definition}{Definition}%
\usepackage{ragged2e}
\begin{document}

\title[{\sdvs} A Post-Quantum Secure Compact Strong-Designated Verifier Signature]{{\sdvs}: A Post-Quantum Secure Compact Strong-Designated Verifier Signature}


\author[1]{\fnm{Shanu} \sur{Poddar}}\email{shanupoddar2112@gmail.com}

\author[1]{\fnm{Sweta} \sur{Mishra}}\email{sweta.mishra@snu.edu.in}

\author[2]{\fnm{Tapaswini} \sur{Mohanty}}\email{mtapaswini37@gmail.com}

\author*[3]{\fnm{Vikas} \sur{Srivastava}}\email{vikas.math123@gmail.com}

\author[2]{\fnm{Sugata} \sur{Gangopadhyay}}\email{sugata.gangopadhyay@cs.iitr.ac.in }

\affil[1]{\orgdiv{Department of Computer Science and Engineering}, \orgname{Shiv Nadar Institute of Eminence, NCR}, \orgaddress{\street{} \city{} \postcode{201314}, \state{Delhi}, \country{India}}}

\affil[2]{\orgdiv{Department of Computer Science and Engineering}, \orgname{Indian Institute of Technology Roorkee}, \orgaddress{\street{}\city{Roorkee}, \postcode{247667}, \state{Uttarakhand}, \country{India}}}

\affil*[3]{\orgdiv{Department of Mathematics}, \orgname{Indian Institute of Technology Madras}, \orgaddress{\street{}\city{Chennai}, \postcode{600036}, \state{Tamil Nadu}, \country{India}}}

\abstract{
Digital signatures are fundamental cryptographic primitives that ensure the authenticity and integrity of digital communication. However, in scenarios involving sensitive interactions---such as e-voting or e-cash---there is a growing need for more controlled signing mechanisms. \emph{Strong-Designated Verifier Signature} (SDVS) offers such control by allowing the signer to specify and restrict the verifier of a signature. The existing state-of-the-art SDVS are mostly based on number-theoretic hardness assumptions. Thus, they are not secure against quantum attacks. Moreover, Post-Quantum Cryptography (PQC)-based SDVS are inefficient and have large key and signature sizes. 
In this work, we address these challenges and propose an efficient post-quantum SDVS (namely, \sdvs) based on ideal lattices under the hardness assumptions of the {\sf Ring-SIS} and {\sf Ring-LWE} problems. \sdvs achieves advanced security properties including strong unforgeability under chosen-message attacks, non-transferability, non-delegatability, and signer anonymity. By employing the algebraic structure of rings and the \emph{gadget trapdoor} mechanism of Micciancio et al., we design \sdvs to minimize computational overhead and significantly reduce key and signature sizes. Notably, our scheme achieves a compact signature size of $\mathcal{O}(n\log q)$, compared to $\mathcal{O}(n^2)$ size, where $n$ is the security parameter, in the existing state-of-the-art PQC designs. To the best of our knowledge, \sdvs offers the \textit{smallest private key and signature size} among the existing PQC-based SDVS schemes.}

\keywords{Lattice-Based Cryptography, Post-Quantum Cryptography, Strong Designated Verifier Signature, Ring-SIS, Ring-LWE}



\maketitle

\section{Introduction}

Digital Signatures \cite{nist1992digital} are a crucial cryptographic primitive that provides security properties such as \textit{authentication}, \textit{integrity}, and \textit{non-repudiation} \cite{katz2010digital}. In particular, it ensures that a message originates from a particular sender, has not been tampered with in between, and prevents any denial caused by signers in digital communication. A digital signature is publicly verifiable using the signer's public key and cannot prevent dishonest verifiers from transferring the validation of sensitive signed messages. However, certain real-life applications exist where the sender of a digital message wants to ensure that only the designated receivers can verify and be convinced whether the signature is valid. As a solution to this problem, the notion of \textit{Designated Verifier Signature} (DVS) was introduced by Jakobsson et al. \cite{10} in Eurocrypt'96. In a DVS scheme, only the designated verifiers can verify the signature and obtain the conviction of correctness of the proof. In this scheme, the signer and the designated verifier both have the equivalent signing privileges, i.e., the designated verifier also has the right to produce a valid signature over the same message. Any third party except the signer and the designated verifier can not distinguish whether the signature was generated by the actual signer or the designated verifier. This property is named as \textit{Non-Transferability} (\textit{NT}). \textit{Strong Designated Verifier Signature} (SDVS), which is a stronger notion of the DVS scheme (also proposed by Jakobsson et al \cite{10} and further explained by Vergnaud et al. \cite{11}) provides some extra useful privacy properties. Specifically, in SDVS, the private key of the designated verifier is also involved while verifying the signature. This means that nobody else (including a third party) can check the validity of the signature, even if they have the public key. SDVS further provides \textit{Non-Delegatability} (\textit{ND}), which is a desired property in cases where the responsibility of the signer becomes very important. For example, an e-voting protocol where it should not be possible to delegate the signing rights to others. Thus, {SDVS} can be useful in many applications, e.g., e-voting, digital subscription system, online contract agreements, product licensing, etc. 

Several SDVS exist in the literature, but the existing state-of-the-art SDVS protocols are based on the number-theoretic assumptions \cite{11}. Quantum algorithms like Shor's algorithm \cite{12} may be used for the cryptanalysis of the existing SDVS designs.  To withstand the quantum threats, there is an urgent need to design a quantum secure SDVS protocol. A new direction of research called \textit{Post-Quantum Cryptography} (PQC) has been announced by the \textit{NIST} in 2016 \cite{18}. The goal of PQC is to design and analyze protocols that can securely withstand quantum attacks. PQC can be divided into several categories, such as Lattice-based \cite{17}, Multivariate-based \cite{ding2006multivariate}, Hash-based \cite{16}, Code-based \cite{overbeck2009code}, and Isogeny-based \cite{mishra2025survey}. Among these, Lattice-based Cryptography is one of the most important research directions in PQC. In this setting, the security of the proposed cryptographic protocols relies on the worst-case hardness of lattice-based problems \cite{13}.\\

\noindent {\it \textbf{Related Works}.}
A lot of fundamental work has been done in the domain lattice-based cryptography over the last two decades. One of the most groundbreaking works was GPV \cite{7} - the first lattice-based provably secure signature scheme. In \cite{7}, the authors introduced a new notion of trapdoor function called \textit{Pre-image Samplable Function} (PSF). Based on the idea of this trapdoor, several lattice-based cryptographic protocols have been defined. 
The first lattice-based DVS scheme was proposed by Wang et al. \cite{2}, which employed PSF and utilized Bonsai Trees for the basis delegation process. Although the scheme was proven secure in the random oracle model, it suffers from large key and signature sizes. In 2017, Noh et al. \cite{14} proposed an SDVS scheme that was proven secure in the standard model. They used the technique of Learning with Errors (LWE)-based public key cryptosystem, and lattice-based chameleon hash function. The SDVS design in \cite{14} also suffered from complex computations and large key and signature sizes. In 2019, Cai et al. \cite{4} proposed a comparatively efficient SDVS scheme based on the hardness of the Ring-SIS problem. They used the idea of rejection sampling to generate a signature of efficient size. \textit{ND} was introduced for SDVS by Lipma et al. \cite{15}. As discussed above, \textit{ND} is a crucial property in real-life applications of SDVS. Earlier constructions of SDVS with \textit{ND} security property were proposed only in the classical setting. In 2024, the first lattice-based SDVS with \textit{ND} property was proposed by Zhang et al. \cite{1}, which provides provable security based on the hardness of SIS and LWE problems. A summary of the related works is provided in Table \ref{tab:lit-survey}.\\


\noindent {\bf Our Contribution:} The major contribution of this work is summarized below. 
\begin{enumerate}
    \item We propose an efficient post-quantum SDVS, namely \sdvs based on ideal lattice under the hardness assumptions of {\sf Ring-SIS} and {\sf Ring-LWE}
    \item \sdvs provides advanced security properties such as strong unforgeability under the chosen-message attack, non-transferability, non-delegatability, and privacy of signer's identity.
    \item We design \sdvs using the algebraic structures of the ring. Thus, we are able to minimize the computational costs, parameter sizes, and signature and keys overhead by a significant margin. We exploited the idea of the trapdoor, called the Gadget trapdoor, defined by Micciancio et al. in \cite{6}. We emulated this trapdoor definition in the ideal lattice computations that gave us more efficiency than the standard lattices based SDVS designs.
    \item \sdvs is very efficient and compact when compared to the existing state-of-the-art SDVS. In particular, \sdvs provides the {\it smallest private key size} among the existing PQC-based SDVS design. 
    
    \item Our scheme achieves a \textit{compact signature size} of $\mathcal{O}(n \log q)$,  
compared to the $\mathcal{O}(n^2)$ signature sizes in standard constructions.  This corresponds to a {reduction by a factor of} $n/ \log q$. 
In fact, the signature size in \sdvs is {\it smallest} among all PQC-based SDVS designs.
\end{enumerate}

\setlength{\tabcolsep}{4pt}
\begin{table}
    \centering
    \small
    \caption{Tecniques, advantages, and limitations of existing lattice-based SDVS schemes}
    \label{tab:lit-survey}
    
    \begin{tabular}{|>{\raggedright\arraybackslash}p{1.2cm}|>{\raggedright\arraybackslash}p{2.3cm}|>{\raggedright\arraybackslash}p{5cm}|>{\raggedright\arraybackslash}p{3cm}|}
        \hline
        \textbf{Scheme} & \textbf{Cryptographic Techniques} & \textbf{Advantages} & \textbf{Limitations} \\
        \hline
        \textbf{Wang et al. } \cite{2} &
        \begin{minipage}[t]{\dimexpr\linewidth-2\tabcolsep\relax}
            $\bullet$ Pre-image Sampling Functions\\
            $\bullet$ Bonsai Trees
        \end{minipage} &
        \begin{minipage}[t]{\dimexpr\linewidth-2\tabcolsep\relax}
            $\bullet$ First lattice-based SDVS scheme\\
            $\bullet$Based on the hardness of LWE and SIS problem\\
            $\bullet${ Provides \textit{EU} and \textit{NT} security}
        \end{minipage} &
        \begin{minipage}[t]{\dimexpr\linewidth-2\tabcolsep\relax}
            $\bullet$ Does not provide \textit{PSI} and \textit{ND} security\\
            $\bullet$ Complex computations\\
            $\bullet$ Large key sizes
        \end{minipage} \\
        \hline
        \textbf{Noh et al. } \cite{14} &
        \begin{minipage}[t]{\dimexpr\linewidth-2\tabcolsep\relax}
            $\bullet$ Used LWE-based PKC\\
            $\bullet$ Chameleon Hash function
        \end{minipage} &
        \begin{minipage}[t]{\dimexpr\linewidth-2\tabcolsep\relax}
            $\bullet$ First lattice-based SDVS scheme in standard model\\
            $\bullet${ Provides \textit{SU}, \textit{NT}, and \textit{PSI} security}\\
            $\bullet$ Based on the hardness of SIS and LWE problems
        \end{minipage} &
        \begin{minipage}[t]{\dimexpr\linewidth-2\tabcolsep\relax}
            $\bullet$ Does not provide \textit{ND} security\\
            $\bullet$ Complex computations\\
            $\bullet$ Very large key sizes
        \end{minipage} \\
        \hline
        \textbf{Cai et al.} \cite{4} &
        \begin{minipage}[t]{\dimexpr\linewidth-2\tabcolsep\relax}
            $\bullet$ Hard problems in ideal lattice\\
            $\bullet$ Filtering techniques
        \end{minipage} &
        \begin{minipage}[t]{\dimexpr\linewidth-2\tabcolsep\relax}
            $\bullet$ Not using PSF and Bonsai Trees\\
            $\bullet$Resisting side Channel attacks\\
            $\bullet${ Provides \textit{EU, NT}, and \textit{PSI} security}\\
            $\bullet${ Based on the hardness of $\mathcal{R}-$SIS problem}
        \end{minipage} &
        \begin{minipage}[t]{\dimexpr\linewidth-2\tabcolsep\relax}
            $\bullet$ Does not offer \textit{ND} security\\
            $\bullet$ Large key sizes, but small signature size.
        \end{minipage} \\
        \hline
        \textbf{Zhang et al.} \cite{1} &
        \begin{minipage}[t]{\dimexpr\linewidth-2\tabcolsep\relax}
            $\bullet$ Rejection sampling\\
            $\bullet$ Pre-image Samplable Functions
        \end{minipage} &
        \begin{minipage}[t]{\dimexpr\linewidth-2\tabcolsep\relax}
            $\bullet$ First lattice-based SDVS with \textit{ND} security\\
            $\bullet$Provable secure based on SIS and LWE problem\\
            $\bullet$Provides \textit{EU, NT}, and \textit{PSI} security
        \end{minipage} &
        \begin{minipage}[t]{\dimexpr\linewidth-2\tabcolsep\relax}
            $\bullet$ Large key and signature sizes\\
            $\bullet$ Impractical for real-world applications
        \end{minipage} \\
        \hline
    \end{tabular}
    {\textit{EU}: Existential Unforgeability, \textit{SU}: Strong Unforgeability, \textit{NT}: Non-Transferability, \textit{PSI}: Privacy for Signer's Identity, \textit{ND}: Non-Delegatability. \textit{SIS}: Short Integer Solution,\textit{LWE}: Learning With Errors, \textit{PKC}: Public Key Cryptosystems.}
\end{table}





\section{Preliminaries}

\subsection{Notations}
The notations, which have been used in this paper, are described in Table~\ref{Table 1}.

\begin{table}[h]
\caption{Notations for the symbols used in the paper}
\label{Table 1}
  \begin{tabular}{|c|c|}
  \hline
   $\mathcal{R}$ & Ring of polynomials of degree $n-1$ $\mathbb{Z}[x]/<x^n+1>$ with integer coefficients \\
   \hline
   $\mathcal{R}_q$ &  Ring of polynomials of degree $n-1$ $\mathbb{Z}_q[x]/<x^n+1>$ with coefficients in $\mathbb{Z}_q$ \\
   \hline
   $\xleftarrow{\$}$ & Sampling uniformly random elements\\
   \hline
   $\mathbf{a}$ & Bold small-case letter denotes a vector \\
   \hline
   $\mathbf{R}$ & Bold capital-case letter denotes a matrix\\
   \hline   
   $t$  & Normal small-case letter denotes an element of the ring \\
   \hline
   $\pmod q$ & Elements from the set $(-(q-1)/2, \ldots, 0, \ldots, (q-1)/2]$\\
   \hline
   $PPT$  & Probabilistic Polynomial Time \\
   \hline
   $||$ & Concatenation \\
   \hline
   $||\cdot||$ & Euclidean norm\\
   \hline
   $\mathbf{a}\cdot \mathbf{b}$ & Inner product of two vectors in the ring\\
   \hline
   $\mathbf{a}t$ & Scalar product between a vector of the ring and a ring element\\
   \hline
   $\emptyset$ & Empty set\\
   \hline
   $\mathbf{a}|\mathbf{b}$ & Concatenation of elements of $\mathbf{a}$ followed by the elements of $\mathbf{b}$\\
 \hline
  \end{tabular}

\end{table}

\subsection{Background of Lattice-Based Cryptography}
In this section, we discuss the concepts of lattices, ideal lattices, and the important results in the ideal lattices.
\begin{definition}
\textbf{(Lattice \cite{8})} Let $\mathbf{v_1},\ldots , \mathbf{v_n} \in \mathbb{R}^m$ be a set of linearly independent vectors. The lattice \(\Lambda\) generated by $\mathbf{v_1}, \ldots , \mathbf{v_n}$ is the set of linear combinations of $\mathbf{v_1},\ldots , \mathbf{v_n}$ with coefficients in integers \(\mathbb{Z}\) i.e; \[\Lambda = \{a_1\mathbf{v_1}+ \ldots  +a_n\mathbf{v_n} : a_1, a_2, \ldots ,a_n \in \mathbb{Z}\}. \]
Three types of integer lattices are mainly considered in the literature. For a given integer modulus $q$, a matrix \(\mathbf{A} \in \mathbb{Z}_q^{n\times m}\), and \(\mathbf{u} \in \mathbb{Z}_q^n\), define:
\[\Lambda_{q}(\mathbf{A}^T) = \{\mathbf{x}\in \mathbb{Z}^m: \exists \; \mathbf{s} \in \mathbb{Z}_q^n \text{ s.t. }\mathbf{A}^T\cdot \mathbf{s} = \mathbf{x}\pmod q\}\]

\[
\Lambda^{\perp}_{q}(\mathbf{A}) = \{\mathbf{x} \in \mathbb{Z}^m : \mathbf{A} \cdot \mathbf{x} \equiv \mathbf{0} \pmod{q} \}
\]

\[\Lambda^{\mathbf{u}}_{q}(\mathbf{A}) = \{\mathbf{x}\in \mathbb{Z}^m: \mathbf{A}\cdot \mathbf{x} = \mathbf{u}\pmod q\}\]
\end{definition}
\noindent We now provide the definition of the discrete Gaussian distribution over a lattice, which plays a central role in the design and analysis of lattice-based cryptographic schemes.
\begin{definition}

\textbf{(Discrete Gaussian \cite{9})} Let \(\Lambda \subset \mathbb{Z}^m, c \in \mathbb{R}^m, \sigma \in \mathbb{R}^+\). Define:
\[\rho_{\sigma, c}(x) = \exp(-\pi\frac{||x-c||^2}{\sigma^2}) \text{ and } \rho_{\sigma, c}(\Lambda) = \sum_{x\in \Lambda}\rho_{\sigma, c}(x).\]The discrete gaussian distribution over $\Lambda$ with center $c$ and parameter $\sigma$ is defined as 
\[\forall x \in \Lambda,\, \mathcal{D}_{\Lambda, \sigma, c}(x) = \frac{\rho_{\sigma, c}(x)}{\rho_{\sigma, c}(\Lambda)}.\]
\end{definition}
We now define the notion of ideal lattices, which are structured lattices arising from polynomial rings and are fundamental to the efficiency and algebraic properties of lattice-based cryptographic constructions.

\begin{definition}
\textbf{(Ideal Lattice \cite{8})} Consider the ring \(\mathcal{R} = \mathbb{Z}[x]/(x^n+1)\) and \(\mathcal{R}_q = \mathbb{Z}_q[x]/(x^n+1)\) for a variable \(x\). Corresponding to the integer lattices defined above, we define three types of lattices over the ring \(\mathcal{R}\). For a integral modulus \(q,\, \mathbf{a} \in \mathcal{R}_q^k \text{ 
 , and  } u \in \mathcal{R}_q\), define:

\[\Lambda_{q}(\mathbf{a}) = \{\mathbf{x}\in \mathcal{R}^k: \exists  s \in \mathcal{R}_q \text{ s.t. } \mathbf{a}\cdot s = \mathbf{x}\pmod q\}\]

\[\Lambda^{\perp}_{q}(\mathbf{a}^T) = \{\mathbf{x}\in \mathcal{R}^k: \mathbf{a}^T\cdot \mathbf{x} = 0\pmod q\}\]

\[\Lambda^{u}_{q}(\mathbf{a}^T) = \{\mathbf{x}\in \mathcal{R}^k: \mathbf{a}^T\cdot \mathbf{x} = u\pmod q\}\]
 \end{definition}

\begin{definition}
(\textbf{Short Integer Solution (SIS) Problem \cite{13}}) Given parameters $n,m, \text{ and } q$; where $n$ is the security parameter, $m = \mathcal{O}(n\log q)$ and $q= poly(n)$ is a prime modulus value. For a uniform random matrix $\mathbf{A} \in \mathbb{Z}_q^{n\times m}$, find a non-zero vector $\mathbf{e}\in \Lambda_q^{\perp}(\mathbf{A})\subset \mathbb{Z}^m$ with ``small norm" s.t. $\mathbf{A}\mathbf{e} = 0 \pmod q$ 
\end{definition}

\begin{definition}
(\textbf{Decisional Learning With Errors (LWE) Problem \cite{13}}) Given parameters $n,m, \text{ and } q$: where $n$ is the security parameter, $m = \mathcal{O}(n\log q)$ and $q = poly(n)$ is a prime modulus value. For a given pair $(\mathbf{A}, \mathbf{b} = \mathbf{A}^T\mathbf{s} + \mathbf{e}) \in \mathbb{Z}_q^{n\times m} \times \mathbb{Z}^m$ s.t. $\mathbf{A}\xleftarrow{\$}\mathbb{Z}^{n\times m}$,  $\mathbf{s}\xleftarrow{\$}\mathbb{Z}^n$ and $\mathbf{e} \in \mathbb{Z}^m$ is coming from a gaussian distribution defined over $\mathbb{Z}^m$; distinguish this pair with a uniformly random pair $(\mathbf{A}, \mathbf{b}) \in \mathbb{Z}^{n\times m} \times \mathbb{Z}^m$.
\end{definition}

\begin{definition}
(\textbf{Ring-SIS Problem \cite{9}}) Given a ring 
$\mathcal{R}_q = \mathbb{Z}_q[x]/<x^n + 1>$ and uniformly random elements $a_1, \dots, a_l \in \mathcal{R}_q$, where $n$ is the security parameter, $q = poly(n)$ a prime modulus, and $l = \mathcal{O}(\log q)$. Find $l$ ``short elements" $\mathbf{e} = \{e_1, \ldots, e_l \}\in \Lambda_q^{\perp}(\mathbf{a^T}) \subset \mathcal{R}^l$ s.t. 
\(a_1e_1 + \ldots + a_le_l = 0 \pmod q\).
\end{definition}

\begin{definition}
(\textbf{Decisional Ring-LWE Problem\cite{9}}) Given \( \mathbf{a} = (a_1, \ldots, a_l)^T \in \mathcal{R}_q^l \), a vector of \( l \) uniformly random polynomials, and \( \mathbf{b} = \mathbf{a}s + \mathbf{e} \pmod q\), where \( s \xleftarrow{\$} \mathcal{R}_q\) and \( \mathbf{e} \xleftarrow{\$} \mathcal{R}_q^l \) from a gaussian distribution defined over $\mathcal{R}_q$; distinguish \( (\mathbf{a}, \mathbf{b} = \mathbf{a}s + \mathbf{e}) \) from \( (\mathbf{a}, \mathbf{b}) \) drawn uniformly at random from \( \mathcal{R}_q^l \times \mathcal{R}_q^l \).

\end{definition}

We now introduce the concept of trapdoors for ideal lattices, which enable efficient preimage sampling and form the foundation of many lattice-based cryptographic constructions, including identification and signature schemes.

\begin{definition}
\textbf{(\textbf{g}-Trapdoor for Ideal Lattice\cite{9})} Let \(\mathbf{a} \in \mathcal{R}_q^{l+k} \text{ and } \mathbf{g} = (1, 2, \ldots, 2^{k-1}) \in \mathcal{R}_q^{k}\). A $\mathbf{g}-$trapdoor for $\mathbf{a}$ is a collection of linearly independent vectors of ring elements \(\mathbf{R} = (\mathbf{r_1, \ldots, r_k}) \in \mathcal{R}_q^{l\times k}\) such that \(\mathbf{a}^T \begin{bmatrix}
\mathbf{R} \\ 
\mathbf{I}_k 
\end{bmatrix}= h\mathbf{g}^T\) for some non-zero ring element \(h\in \mathcal{R}_q\). The value \( h \) is known as the tag of the trapdoor. The effectiveness of the trapdoor is evaluated based on its largest singular value, \( s_1(\mathbf{R}) \), which is determined as the maximum singular value of the matrix representation of \( \mathbf{R} \) in \( \mathbb{Z}_q^{ln \times ln} \).
\end{definition}

\begin{lemma}[$(\mathbf{a}, \mathbf{R})\leftarrow${\sf ringGenTrap}$^{D}(\mathbf{a}_0, h)$ \cite{8}]
Given as input a vector of ring elements\break $\mathbf{a}_0 = (a_1, \dots, a_l)^T \in \mathcal{R}_q^l$, a non-zero invertible ring element $h \in \mathcal{R}_q$, a distribution $\chi^{l \times k}$ over $\mathcal{R}^{l \times k}$. (If no particular $\mathbf{a}_0, h$ are given as input, then the algorithm may choose them itself, e.g., picking $\mathbf{a}_0 \leftarrow \mathcal{R}_q^l$ uniformly, and setting $h = 1$.), the algorithm outputs a vector of ring elements $\mathbf{a} = (\mathbf{a}_0^T, \mathbf{a}_1^T = h\mathbf{g}^T - \mathbf{a}_0^T\mathbf{R})^T \in \mathcal{R}_q^{l+k}$, and a trapdoor $\mathbf{R} = (\mathbf{r}_1, \dots, \mathbf{r}_k) \in \mathcal{R}^{l \times k}$ with tag $h \in \mathcal{R}_q$. Moreover, the distribution of $\mathbf{a}$ is close to uniform (either statistically or computationally) as long as the distribution of $(\mathbf{a}_0^T, -\mathbf{a}_0^T \mathbf{R})$ is.
\end{lemma}

\begin{lemma}[{\sf Regularity Lemma} \cite{8}]
Let \(a_i \xleftarrow{\$} \mathcal{R}_q \text{ and } r_i\xleftarrow{\$} \chi \text{ for } i= 1, \ldots, l\), where $\chi$ is a distribution over $\mathcal{R}_q$. Then we obtain that the statistical distance between \(\mathbf{b} = \sum_{i=1}^{l} a_i r_i\) and the uniform distribution over $\mathcal{R}_q$ is $2^{-\Omega(n)}$. We can further extend this lemma for a vector of elements from $\mathcal{R}_q$.
\end{lemma}

\begin{lemma}[{\sf ringInvert}$^{\mathcal{O}}(\mathbf{R}, \mathbf{a}, \mathbf{b})$ \cite{8}] 
Given as input an oracle $\mathcal{O}$ for inverting the function $\alpha_{\mathbf{g}}(s', \mathbf{e}') = \mathbf{g}s' + \mathbf{e'} \pmod q$ where $\mathbf{e}' \in \mathcal{R}^k$ is suitably small from a gaussian like distribution and $s' \in \mathcal{R}_q$, a vector of ring elements $\mathbf{a} \in \mathcal{R}_q^{l+k}$, a $\mathbf{g}$-trapdoor $\mathbf{R} \in R^{l \times k}$ for $\mathbf{a}$ with tag $h$, a vector $\mathbf{b} = \mathbf{a}s + \mathbf{e} \pmod q$ for any random $s \in \mathcal{R}_q$ and suitably small $\mathbf{e} \in \mathcal{R}^{l+k}$ coming from a narrow distribution over $\mathcal{R}_q^{l+k}$, the algorithm outputs $s$ and $\mathbf{e}$. 
\end{lemma}

Given a trapdoor $\mathbf{R}$ for $\mathbf{a} \in \mathcal{R}_q^{l+k} \text{ and } u = \mathbf{a^T \cdot x} \pmod q$, the sampling algorithm in the following lemma finds the solution $\mathbf{x}$ from desired distribution. 

\begin{lemma}\label{lemma-ring-sample}({\sf ringSample}$^{\mathcal{O}}(\mathbf{R}, \mathbf{a_0}, h,  u, \sigma)$ \cite{8}) Given as input in offline mode (i) an oracle $\mathcal{O}(v)$ for gaussian sampling over a desired coset \(\Lambda^{v}_{q}(\mathbf{g}^T)\) with parameter $\sigma$, where $v\in \mathcal{R}_q$, (ii) a vector of ring elements $\mathbf{a_0} \in \mathcal{R}_q^l$, (iii) a trapdoor $\mathbf{R} \in \mathcal{R}^{l\times k}$, (iv) a gaussian parameter $\sigma$. In addition, given as input in online phase, (i) a non-zero tag $h\in \mathcal{R}_q$ defining \(\mathbf{a} = (\mathbf{a_0}^T, h\mathbf{g}^T - \mathbf{a}^T_0\mathbf{R})^T \in \mathcal{R}_q^{l+k}\), and (ii) a syndrome $u\in \mathcal{R}_q$, the algorithm outputs a vector $\mathbf{x}$ drawn from a distribution statistically close to $\mathcal{D}_{\Lambda^v_q(\mathbf{a_0}^T), \sigma'}$ for some Gaussian parameter $\sigma'$.
\end{lemma}




\subsection{Definition of Strong Designated Verifier Signature (SDVS)}
A SDVS is a collection of the following algorithms. We follow the syntax and definition of SDVS from \cite{1}. Refer to Figure \ref{fig:sdvs} for an illustrative summary.

\begin{figure}
    \centering
    \includegraphics[width=0.95\linewidth]{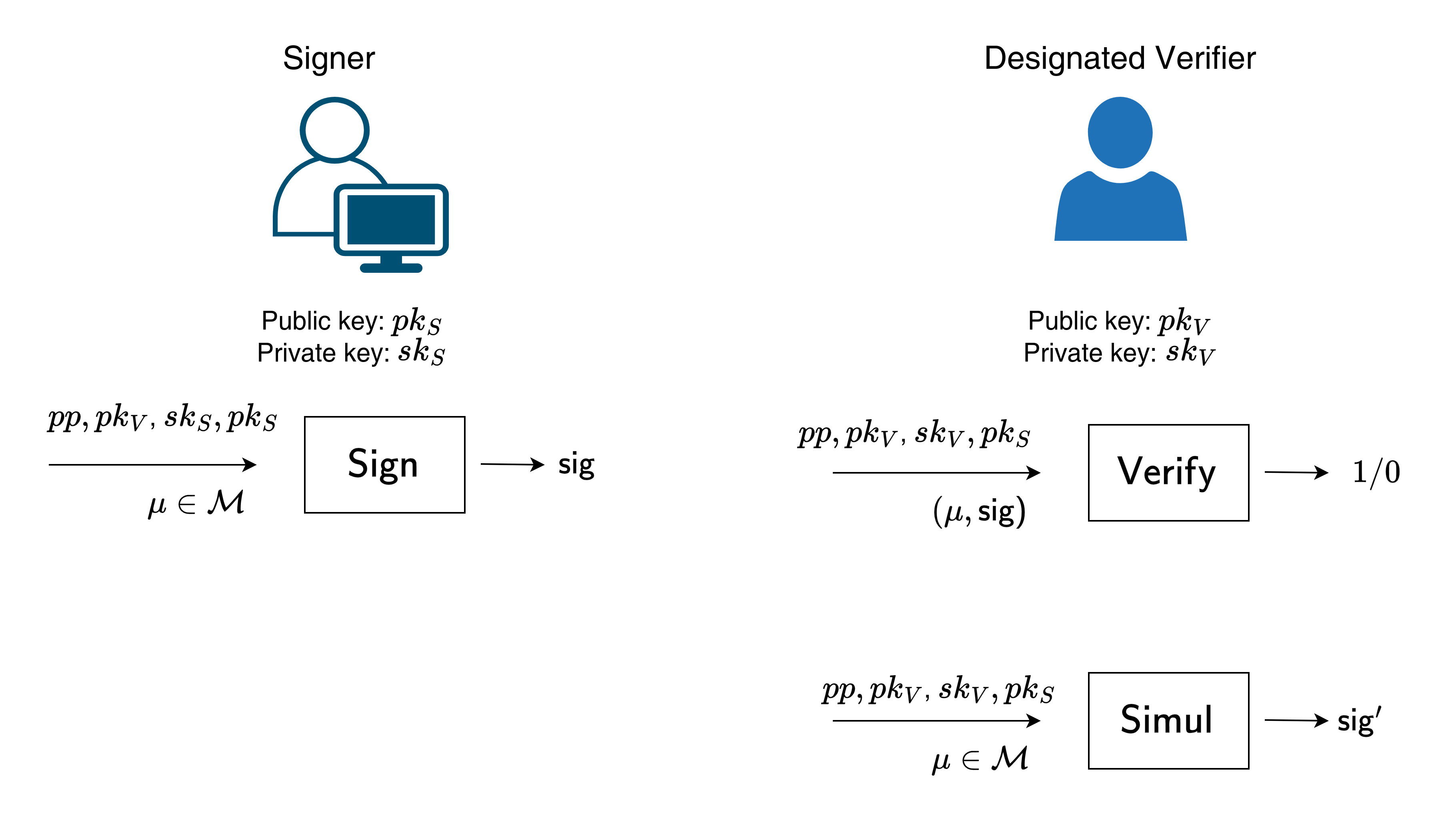}
    \caption{High level overview of SDVS}
    \label{fig:sdvs}
\end{figure}
\begin{description}
    \item[{\sf Setup}$(1^n)$:] It is a probabilistic algorithm. On input a security parameter $n$, it outputs the public parameters $pp$.
    
    \item[{\sf SigKeyGen}$(pp)$:] It is a probabilistic (or deterministic) algorithm. On input $pp$, it outputs the public key $pk_S$ and private key $sk_S$ for a signer $S$.
    
    \item[{\sf VerKeyGen}$(pp)$:] It is a probabilistic (or deterministic) algorithm. On input $pp$, it outputs the public key $pk_V$ and private key $sk_V$ for a designated verifier $V$.
    
    \item[\sf {Sign}$(pp, sk_S, pk_S, pk_V, \mu)$:] It is a probabilistic algorithm. On input $pp$, a private key $sk_S$, the public keys $pk_S$ and $pk_V$ of the signer $S$ and a designated verifier $V$ and a message $\mu \in \mathcal{M}$, it outputs a real designated verifier signature ${\sf sig} \in \mathcal{S}$.
    
    \item[{\sf Verify}$(pp, sk_V, pk_S, pk_V, {\sf sig}, \mu)$:] It is a deterministic algorithm. On input $pp$, a private key $sk_V$, the public keys $pk_S$ and $pk_V$ of the signer $S$ and a designated verifier $V$, a message $\mu \in \mathcal{M}$, and a signature $\sf sig \in \mathcal{S}$, it outputs a boolean decision $b$: $b = 1$ denotes accepting or $b = 0$ denotes rejecting it.
   
    \item[{\sf Simul}$(pp, sk_V, pk_S, pk_V, \mu)$:] It is a probabilistic algorithm. On input $pp$, a private key $sk_V$, the public keys $pk_S$ and $pk_V$ of the signer $S$ and a designated verifier $V$, and a message $\mu \in \mathcal{M}$, it outputs a simulated designated verifier signature $\sf sig' \in \mathcal{S}$.
\end{description}

\subsection{Security Model}
An SDVS scheme must satisfy unforgeability over chosen message attack, non-transferability, privacy of the signer's identity, and non-delegatability. The security definitions are given below.

\subsubsection*{Strong Unforgeable Chosen Message Attack (\textit{SU-CMA}) Security\label{sec:su-cma}}

For our SDVS scheme, we define strong unforgeability, which is a stronger notion than existential unforgeability. It states that any PPT adversary, without the private key of either the signer or the designated verifier, can not provide a valid designated verifier signature for a message that has even been queried in either of the signing or simulating queries. We define this unforgeable security, where the attacker can choose the message in an adaptive manner, by the following game between a challenger \(\mathcal{C}\) and a PPT adversary \(\mathcal{A}\):

\begin{description}
    \item[{\sf Setup}:] \(\mathcal{C}\) runs \({\sf Setup}(1^n)\) to get \(pp\), \({\sf SignKeyGen}(pp)\), and \({\sf VerKeyGen}(pp)\) to get \((pk_S, sk_S)\) and \((pk_V, sk_V)\) for the signer \textit{S} and the designated verifier \textit{V}. \(\mathcal{C}\) keeps  \((sk_S, sk_V)\) in secret and sends \((pp, pk_S, pk_V)\) to \(\mathcal{A}\).
    
    \item[{\sf Signing Queries}:] \(\mathcal{A}\) chooses a message \(\mu_i\) adaptively, \(\mathcal{C}\) runs \({\sf Sign}(pp, sk_S, pk_S, pk_V \allowbreak , \mu_i)\) to obtain a signature \({\sf sig}_i\) and sends it to \(\mathcal{A}\).
    
    \item[{\sf Simulating Queries}:] \(\mathcal{A}\) chooses a message \(\mu_i\) adaptively, \(\mathcal{C}\) runs \({\sf Simul}(pp, sk_V, pk_S \allowbreak, pk_V, \mu_i)\) to obtain a signature \({\sf sig}'_i\) and sends it to \(\mathcal{A}\). \(\mathcal{A}\) is allowed to query the signing and the simulating oracle at most \(q_s = poly(n)\) times.

    \item[\sf {Verification Queries}:]  \(\mathcal{A}\) chooses a message-signature pair \((\mu_i, {\sf sig}_i)\) adaptively, \(\mathcal{C}\) runs \break \({\sf Verify}(pp, sk_V, pk_S, pk_V,{\sf sig}_i, \mu_i)\) to obtain a decisional value 1 for valid and 0 otherwise, and sends it to \(\mathcal{A}\). \(\mathcal{A}\) is allowed to query the verification oracle at most \(q_v = poly(n)\) times.

    \item[{\sf Output}:] \(\mathcal{A}\) outputs a message-signature pair \((\mu^*, {\sf sig}^*)\), and wins if the following conditions hold-
      \begin{itemize}
          \item \({\sf Verify}(pp, sk_V, pk_S, pk_V,{\sf sig}^*, \mu^*) = 1\).
          \item \(({\sf sig}^*, \mu^*) \neq ({\sf sig}_i, \mu_i)\), \(\forall i \in \{1, 2, \ldots, q_s\}\).
      \end{itemize}
\end{description}

The advantage of \(\mathcal{A}\) in the above game, i.e., the probability of $\mathcal{A}$ wins, is denoted by $Adv_{SDVS, \mathcal{A}}^{SU-CMA}(n)$, and the probability is taken over the randomness of \(\mathcal{A}\), the algorithms \({\sf Sign, Simul}\), and \({\sf Verify}\). We say that the SDVS scheme is unforgeable if $Adv_{SDVS, \mathcal{A}}^{SU-CMA}(n)$ is negligible in the security parameter \(n\).

\subsubsection*{Non-Transferability (\textit{NT}) Security}
In the context of the SDVS scheme, non-transferability means that the designated verifier should not be able to transfer the conviction of the validity of a real designated verifier signature. This is accomplished by the \({\sf Sign, Simul}\) algorithm, using which the designated verifier can generate signatures that are indistinguishable from those generated by the actual signer. The \textit{NT} security of our SDVS is defined as follows:
\begin{description}
    \item[{\sf Setup}:] The challenger \(\mathcal{C}\) runs \({\sf Setup}(1^n)\) to get \(pp\), then runs \({\sf SignKeyGen}(pp)\), and \({\sf VerKeyGen}(pp)\) to get \((pk_S, sk_S)\) and \((pk_V, sk_V)\) for the signer \textit{S} and the designated verifier \textit{V}. {\(\mathcal{C}\) sends  \(( pk_S, pk_V)\) to \(\mathcal{A}\)}.
    
    \item[{\sf Challenge}:] \(\mathcal{A}\) chooses a message \(\mu^*\) adaptively, \(\mathcal{C}\) runs \({\sf Sign}(pp, sk_S, pk_S, pk_V, \mu^*)\) to obtain a signature \({\sf sig}^*_0\), then runs \({\sf Simul}(pp, sk_V, pk_S, pk_V, \mu^*)\) to obtain a signature \({\sf sig}^*_1\). \(\mathcal{C}\) chooses a random bit \(b\in  \{0,1\}\) and sends \({\sf sig}^*_b\) to \(\mathcal{A}\).
    
    \item[{\sf Output}:] \(\mathcal{A}\) outputs a bit \(b^*\), and wins the challenge if \(b^* = b\).
\end{description}

The advantage of \(\mathcal{A}\) in the above game is defined by $Adv_{SDVS, \mathcal{A}}^{NT}(n) = |Pr[b^* = b]- 1/2|$, and the probability is taken over the randomness of \(\mathcal{A}\), the algorithms \({\sf Sign}\) and \({\sf Simul}\). We say that the SDVS scheme is non-transferable if $Adv_{SDVS, \mathcal{A}}^{NT}(n)$ is negligible in the security parameter \(n\).

\subsubsection*{Privacy of Signer's Identity (\textit{PSI}) Security}
\textit{PSI} security provides privacy to the identity of the signer. Given a designated verifier signature and two valid signing public keys, any adversary/eavesdropper in between the signer and the designated verifier can not determine which private key, corresponding to the given public keys, has been used to create the signature with a non-negligible probability. This feature is, precisely, made possible by involving the private key of the designated verifier in the verification process. The \textit{PSI} security of the SDVS scheme is defined by the following game between a PPT adversary \(\mathcal{A}\) and a challenger \(\mathcal{C}\):
\begin{description}
    \item[{\sf Setup}:] \(\mathcal{C}\) runs \({\sf Setup}(1^n)\) to get \(pp\), runs \({\sf SignKeyGen}(pp)\) twice to get \((pk_{S_0}, sk_{S_0})\), \((pk_{S_1}, sk_{S_1})\), and runs \({\sf VerKeyGen}(pp)\) to get \((pk_V, sk_V)\) for the signers $S_0, S_1$ respectively and the designated verifier \textit{V}. \(\mathcal{C}\) keeps  \(sk_V\) in secret and sends \((pp, pk_{S_0}, sk_{S_0}, pk_{S_1}, sk_{S_1}, pk_V)\) to \(\mathcal{A}\).
    
    \item[{\sf Simulating Queries}:] \(\mathcal{A}\) chooses a message \(\mu_i\) adaptively and a bit \(b\in \{0,1\}\), \(\mathcal{C}\) runs \break \({\sf Simul}(pp, sk_V, pk_{S_b}, pk_V, \mu_i)\) to obtain a signature \({\sf sig}'_{i,b}\) and sends it to \(\mathcal{A}\). \(\mathcal{A}\) is allowed to query the simulating oracle at most \(q_s = poly(n)\) times.
    
    \item[\sf {Verification Queries}:] \(\mathcal{A}\) chooses a message-signature pair \((\mu_i, {\sf sig}_i)\) and a bit \(b\in \{0,1\}\) adaptively,  \(\mathcal{C}\) runs  \({\sf Verify}(pp, sk_V, pk_{S_b}, pk_V,{\sf sig}_i, \mu_i)\) to obtain a decisional value 1 for valid and 0 otherwise, and sends it to \(\mathcal{A}\). \(\mathcal{A}\) is allowed to query the verification oracle at most \(q_v = poly(n)\) times.

    \item[{\sf Challenge}:] \(\mathcal{A}\) chooses a message \(\mu^*\), \(\mathcal{C}\) runs \({\sf Sign}(pp, sk_{S_0}, pk_{S_0}, pk_V, \mu^*)\) to obtain a signature \({\sf sig}_0^*\), then runs \({\sf Sign}(pp, sk_{S_1}, pk_{S_1}, pk_V, \mu^*)\) to obtain a signature \({\sf sig}_1^*\), chooses a random bit \(b\in \{0,1\}\), and sends \({\sf sig}_b^*\) to \(\mathcal{A}\).
    
    \item[{\sf Output}:] \(\mathcal{A}\) outputs a bit \(b^* \in \{0,1\}\), and wins if the following conditions hold-
      \begin{itemize}
          \item \(b^* = b\),.
          \item \(({\sf sig_b}^*, \mu^*) \) was not queried in the verification queries for \(b^* \in \{0,1\}\).
      \end{itemize}
\end{description}
The advantage of \(\mathcal{A}\) is defined as \(Adv_{SDVS, \mathcal{A}}^{PSI}(n) = |Pr[b^* = b] - 1/2\), and we say that the SDVS scheme holds \textit{PSI} security if \(Adv_{SDVS, \mathcal{A}}^{PSI}(n)\) is negligible in \(n\).

\subsubsection*{Non-Delegatability (\textit{ND}) Security}
In SDVS scheme, \textit{ND} security tells that without knowing the private key of either the signer or the designated verifier, no one can generates a valid designated verifier signature. In other words, a \textit{ND} secure SDVS accomplish the proof of knowledge of either signer's private key or designated verifier's private key in a non-interactive manner. This security is defined in the form of a game between an extractor \(\mathcal{E}\) and a black-box \(\mathcal{A}\) that produces a valid signature:

\begin{description}
    \item[{\sf Setup}:] The extractor \(\mathcal{E}\) runs \({\sf Setup}(1^n)\) to get \(pp\) and sends it to \(\mathcal{A}\). \(\mathcal{A}\) then sends either \(pk_S\) or \(pk_V\) to \(\mathcal{E}\).

    \item[{\sf --}] If \(\mathcal{A}\) sends \(pk_S\) to \(\mathcal{E}\), then the game proceeds as follows-

    \item[{\sf VerKeyGen}:] \(\mathcal{E}\) runs \({\sf VerKeyGen}(pp)\) to obtain \((pk_V, sk_V)\), keeps \(sk_V\) in secret and sends \(pk_V\) to \(\mathcal{A}\).
    
    \item[{\sf Simulating Queries}:] \(\mathcal{A}\) chooses a message \(\mu_i\) adaptively , \(\mathcal{E}\) runs \({\sf Simul}(pp, sk_V, pk_{S} \allowbreak , pk_V, \mu_i)\) to obtain a signature \({\sf sig}_{i}\) and sends it to \(\mathcal{A}\). \(\mathcal{A}\) is allowed to query the simulating oracle at most \(q_s = poly(n)\) times.
    
    \item[\sf {Verification Queries}:] \(\mathcal{A}\) chooses a message-signature pair \((\mu_i, {\sf sig}_i)\) adaptively, \(\mathcal{E}\) runs\break \({\sf Verify}(pp, sk_V, pk_{S}, pk_V,{\sf sig}_i, \mu_i)\) to obtain a decisional value, 1 for valid and 0 otherwise, and sends it to \(\mathcal{A}\). \(\mathcal{A}\) is allowed to query the verification oracle at most \(q_v = poly(n)\) times.

    \item[{\sf Challenge}:] \(\mathcal{E}\) chooses  a message \(\mu^*\) and \(\mathcal{A}\) sends a corresponding designated verifier signature \({\sf sig}^*\) to \(\mathcal{E}\).
    
    \item[{\sf Output}:] \(\mathcal{E}\) outputs the private key \(sk_S\) of the signer.

    \item[{\sf --}] If \(\mathcal{A}\) sends \(pk_V\) to \(\mathcal{E}\), then the game proceeds as follows-

    \item[{\sf SigKeyGen}:] \(\mathcal{E}\) runs \({\sf SigKeyGen}(pp)\) to obtain \((pk_S, sk_S)\), keeps \(sk_S\) in secret and sends \(pk_S\) to \(\mathcal{A}\).

    \item[\sf {Signing Queries}:] \(\mathcal{A}\) chooses a message \(\mu_i\) adaptively, \(\mathcal{E}\) runs \({\sf Sign}(pp, sk_S, pk_{S}, pk_V, \mu_i)\) to obtain a signature  \({\sf sig}_i\) and sends it to \(\mathcal{A}\). \(\mathcal{A}\) is allowed to query the signing oracle at most \(q_s = poly(n)\) times.

    \item[{\sf Challenge}:] \(\mathcal{E}\) chooses  a message \(\mu^*\) and \(\mathcal{A}\) sends a corresponding designated verifier signature \({\sf sig}^*\) to \(\mathcal{E}\).

    \item[{\sf Output}:] \(\mathcal{E}\) outputs the private key \(sk_V\) of the designated verifier.

\end{description}

The advantage of \(\mathcal{E}\) in the above game, in time \(t\)  is denoted by $Adv_{SDVS, \mathcal{E}}^{ND}(n)$, and the probability is taken over the randomness of \(\mathcal{A}\), the algorithms \({\sf Sign, Simul}\), and \({\sf Verify}\). We say that the SDVS scheme is non-delegatable against \(\mathcal{A}\), if $Adv_{SDVS, \mathcal{E}}^{ND}(n) \geq poly(\epsilon')$ and \(t<poly(t')\) where \(\mathcal{A}\) can produce a designated verifier signature in time \(t'\) with a probability \(\epsilon'\).

\section{Proposed Ideal Lattice-based SDVS Scheme}
\noindent {\it High Level Overview.} In this section, we describe the proposed construction (called \sdvs) in detail. We first give a high-level overview of the design. \sdvs is a privacy-preserving digital signature protocol built upon the hardness of mathematical problems over \emph{ideal lattices}. Unlike traditional signature schemes, \sdvs ensures that only a designated verifier can check the validity of a signature, while also allowing the verifier to simulate signatures that are computationally indistinguishable from real ones. \sdvs consists of six algorithms: ({\sf Setup} \textsf{SigKeyGen}, \textsf{VerKeyGen}, \textsf{Sign}, \textsf{Verify}, \textsf{Simul}). In the \textsf{Setup} phase, public parameters are initialized based on underlying ring structures and hash function instantiation. The algorithms \textsf{SigKeyGen} and \textsf{VerKeyGen} generate key pairs for the signer and the verifier, respectively. The signing algorithm \textsf{Sign} enables the signer to produce a short signature for a message, while the \textsf{Verify} algorithm utilizes the trapdoor to validate the signature’s correctness. The \textsf{Simul} algorithm allows the verifier to simulate signatures without interacting with the signer, thus providing non-transferability. We provide a detailed description of the algorithms below. 
\begin{description}
    \item[{\sf Setup}:] Let $n$ be a security parameter, and let $q=poly(n)$ be the underlying ring modulus. Let $k$ and $\kappa$ denote the hash parameters such that $k=\lceil\log q\rceil$  and $2^\kappa \binom{k}{\kappa}\geq 2^{100}$. Let {$d=q^{1/\gamma}$} with $\gamma>1$, $\sigma=\mathcal{O}(\sqrt{\log n})$, and $\chi=\mathcal{D}_{\mathcal{R},\sigma}$. $\mathcal{R}=\mathbb{Z}[x]/<x^n+1>$. $\mathcal{R}_q=\mathbb{Z}_q[x]/<x^n+1>$. $l$ is chosen suct that $l+k = \mathcal{O}(\log q)$. Let $\eta$ be a constant positive real value. Then, the \({\sf Setup}\) phase is executed as follows.
    \begin{itemize}
        \item Sample $\mathbf{a}\xleftarrow{\$} \mathcal{R}_q^{l+k}$.
        \item Select a collision resistant hash function $\mathcal{H}: \{0,1\}^* \rightarrow \{c, c\in \mathcal{R}_q, 0<||c||<\sqrt{\kappa}\}$.
    \end{itemize}
The public parameter ${ pp}$ is set to be $\{n, q, k, \kappa, d, \sigma, \chi, \mathcal{R}, \mathcal{R}_q,l, \mathbf{a}, \mathcal{H}\}$ 
    \item[ ($pk_{S}, sk_{S}$)$ \leftarrow$ {\sf SignKeyGen}({ pp}):]
    On input the public parameters { pp}, the following steps are performed to generate the public key and secret key for the signer.
    \begin{itemize}
        \item  Sample $\mathbf{s} \xleftarrow{\$} \{-d, \ldots,0,\ldots,d\}^{l+k}$
        \item Define $t =\mathbf{a}\cdot \mathbf{s} \pmod q\in \mathcal{R}_q$ 
        \item Output $pk_{S}=t$ and $sk_{S}=\mathbf{s}$
    \end{itemize}
   \item[($pk_{V}, sk_{V}$)$ \leftarrow$ {\sf VerKeyGen}({ pp}):] Similar to the previous algorithm, on input the public parameters { pp}, the following operations are executed to produce the public key and secret key for the designated verifier.
    \begin{itemize}
        \item  Sample $\mathbf{\tilde{b}_0, \tilde{b}_1} \xleftarrow{\$} \mathcal{R}_q^{l}$
        \item Execute {\sf ringGenTrap} with inputs $\mathbf{\tilde{b}_0}$ and $h_0=1$ to output ${\mathbf{{b}_0}}\in \mathcal{R}_q^{l+k}$ and a trapdoor $\mathbf{R}_{\mathbf{{b}_0}} \in \mathcal{R}_q^{l\times k}$.
         \item Execute {\sf ringGenTrap} with inputs $\mathbf{\tilde{b}_1}$ and $h_1=1$ to output ${\mathbf{{b}_1}}\in \mathcal{R}_q^{l+k}$ and a trapdoor $\mathbf{R}_{\mathbf{{b}_1}} \in \mathcal{R}_q^{l\times k}$.
         \item Output $pk_V=(\mathbf{{b}_0}, \mathbf{{b}_1})$ and $sk_V=(\mathbf{R}_{\mathbf{{b}_0}}, \mathbf{R}_{\mathbf{{b}_1}})$
    \end{itemize}
    \item[{\sf sig} $\leftarrow$ {\sf Sign}$({pp}, sk_{S}, pk_{S}, pk_V, \mu\in \mathcal{M}):$] On input ${pp}$, $sk_{S}, pk_{S}$, and $pk_V$, the signer employs the algorithm {\sf Sign} to generate a signature on the message $\mu$ as follows.
     \begin{itemize}
         \item Samples $s\xleftarrow{\$} \mathcal{R}_q$, $\mathbf{e}, \mathbf{y} \xleftarrow{\$} \mathcal{D}_{\mathcal{R}_q^{l+k},\sigma}$
         \item Computes  $\mathbf{c_0}={\mathbf{{b}_0}} \cdot s + \mathbf{e} \in \mathcal{R}_q^{l+k} $
         \item Computes $c_1 =\mathcal{H}(\mathbf{a}\cdot \mathbf{y} + \mathbf{{b}_1} \cdot \mathbf{e}|| t || s|| \mu) \in \mathcal{R}_q$
         \item $\mathbf{z}=\mathbf{s} \cdot c_1+\mathbf{y} \pmod q \in \mathcal{R}_q^{l+k}$
         \item Output ${\sf sig}=(\mathbf{c_0}, c_1, \mathbf{z})$ 
     \end{itemize}
     \item[ $1/0$ $\leftarrow$ {\sf Verify}$(pp, sk_V, pk_V, pk_S, {\sf sig}, \mu)$:] On input ${pp}$, $sk_V$, $pk_S$ and $pk_V$, the  {\sf Verify} outputs $1$ if ${\sf sig}$ is a valid signature on the message $\mu$; otherwise outputs $0$.
       \begin{itemize}
         \item Parse {\sf sig}
         \item Check $0< ||\mathbf{z}||< \eta \sigma\sqrt{l+k}$
        \item Execute {\sf ringInvert} with input $\mathbf{R}_{\mathbf{{b}_0}}$, $\mathbf{c_0}, \mathbf{b_0}$ to output $s, \mathbf{e}$.
        \item  Check whether 
        $c_1=\mathcal{H}({\mathbf{a}\cdot \mathbf{z}-\mathbf{t}c_1+\mathbf{b_1}\cdot \mathbf{e} || \mathbf{t} || s || \mu)}$
        \item  Output $1$ if all the above checks are satisfied, otherwise outputs $0$. 
       \end{itemize}
    \item[{\sf sig'} $\leftarrow$ {\sf Simul}$({pp}, sk_V, pk_V, pk_S, \mu):$] On input ${ pp}$, $sk_V$, $pk_S$ and $pk_V$, a designated verifier utilizes {\sf Simul} to output a simulated signature {\sf sig'} on the message $\mu$
    \begin{itemize}
            \item Sample $s', u'\xleftarrow{\$} \mathcal{R}_q, \mathbf{z'} \xleftarrow{\$} \mathcal{D}_{\mathcal{R}_q^{l+k}, \sigma}$
            \item Let $c'_1 = \mathcal{H}(u'|| t || s' || \mu) \in \mathcal{R}_q$
            \item Run ${\sf RingSample}(\mathbf{b_1}, h_1, \mathbf{R_{b_1}}, u'-\mathbf{a \cdot z} + \mathbf{t}\cdot c'_1, \sigma)$ to get a short $\mathbf{e'} \in \mathcal{D}_{\mathcal{R}_q^{l+k},\sigma}$.
            \item Define $\mathbf{c'_0} = \mathbf{b_0}\cdot s' + \mathbf{e'}$ mod $q$.
            \item Output ${\sf sig'} = (\mathbf{c'_0}, c'_1, \mathbf{z'})$ using rejection sampling technique.
        \end{itemize}
        
\end{description}

\subsubsection*{Correctness}
In the verification phase, the verifier checks the validity of the given message-signature pair $(\mu, {\sf sig})$ by taking $({pp}, sk_V, pk_V, pk_S)$ as input. The correctness analysis is given below.
\noindent If the signature is coming from the actual signer i.e., ${\sf sig} = (\mathbf{c_0}, c_1, \mathbf{z})$, then the verifier first checks if $0< ||\mathbf{z}||< \eta \sigma\sqrt{l+k}$. Then using the subroutine ${\sf ringInvert}$, it obtains $s, \mathbf{e} \leftarrow {\sf ringInvert}(\mathbf{R_{b_0}}, \mathbf{b_0}, \mathbf{c_0})$. Finally, the verifier computes 
   \begin{align}
    c_1^* &= \mathcal{H}(\mathbf{a \cdot z}- \mathbf{t}c_1 + \mathbf{b_1 \cdot e}|| \mathbf{t} || s || \mu) \notag \\
          &= \mathcal{H}(\mathbf{a} \cdot (\mathbf{s}c_1 + \mathbf{y})- \mathbf{a \cdot s}c_1 + \mathbf{b_1 \cdot e} || \mathbf{t} || s || \mu)  \notag \\
          &= \mathcal{H}(\mathbf{a} \cdot \mathbf{s}c_1 + \mathbf{a\cdot y}- \mathbf{a \cdot s}c_1 + \mathbf{b_1 \cdot e} || \mathbf{t} || s || \mu) \notag \\
          &= \mathcal{H}(\mathbf{a\cdot y} + \mathbf{b_1 \cdot e} || \mathbf{t} || s || \mu) \notag \\
          &= c_1 \notag
    \end{align}
    
\noindent If the signature is generated from the ${\sf Simul}$ algorithm by the designated verifier  i.e., ${\sf sig'} = (\mathbf{c'_0}, c'_1, \mathbf{z'})$ then the verifier first checks whether $0< ||\mathbf{z'}||< \eta \sigma\sqrt{l+k}$ and then computes $s', \mathbf{e'} \leftarrow {\sf ringInvert}(\mathbf{R_{b_0}}, \mathbf{b_0}, \mathbf{c'_0})$. Now, from Lemma \ref{lemma-ring-sample}, we get that $\mathbf{b_1\cdot e'}= u'-\mathbf{a\cdot z + t}c'_1$ mod $q$. In the end, the verifier computes,
    \begin{align}
    \tilde{c'}_1 &= \mathcal{H}(\mathbf{a \cdot z'}- \mathbf{t}c'_1 + \mathbf{b_1 \cdot e'} || \mathbf{t} || s' || \mu) \notag \\
          &= \mathcal{H}(\mathbf{a \cdot z'}- \mathbf{t}c'_1 + u'-\mathbf{a\cdot z' + t}c'_1 || \mathbf{t} || s' || \mu)  \notag \\
          &= \mathcal{H}(u' || \mathbf{t} || s' || \mu)  \notag \\
          &= c'_1 \notag
    \end{align}

\section{Efficiency Analysis}
In this section, we present a comprehensive comparative analysis of \sdvs with the existing state-of-the-art PQC-based SDVS. We first present the communication and storage overhead of \sdvs.
\begin{description}
    \item[-] The private key $sk_S$ of the signer is of bit size $(l+k)\log(2d+1) = \mathcal{O}(\log q)$.
    \item[-] The private key $sk_V$ of the designated verifier is of bit size $2(l\times k)(n\log q) = 2n(\log q)^2 = \mathcal{O}(n\log q)$.
    \item[-] The signature ${\sf sig} = (\mathbf{c_0, c_1, z})$ is of bit size $(l+k)n \log q + n\log q + (l+k)n \log q = 2(l+k) n\log q + n\log q = \mathcal{O}(n\log q)$.
\end{description}
We now compare \sdvs with the existing lattice-based SDVS schemes proposed in \cite{2,14,4,1}. The comparative analysis considers sizes of the secret key of the signer and verifier, signature sizes. In addition, we also provided an assessment of security properties such as {\it SU-CMA, NT, PSI}, and {\it ND}. The results of the comparison is provided in Table \ref{tab:compare}. As we can see from the results mentioned in Table \ref{tab:compare}, \sdvs outperforms existing schemes in the communication and storage overhead, while attaining all the desired security properties. Notably, \sdvs achieves {\it SU-CMA, NT, PSI} and {\it ND} security properties, similar to \cite{1}, but with significantly lower keys and signature sizes. Specifically, \textsf{LaSDVS} attains a key and signature size of $\mathcal{O}(\log q)$ and $\mathcal{O}(n \log q)$, respectively, in contrast to earlier works whose sizes are polynomial in the security parameter $n$. This efficiency is primarily due to two design choices; (i) the use of a gadget-based trapdoor~\cite{6} in the ideal lattice setting, and (ii) the avoidance of the traditional GPV trapdoor~\cite{7} used in standard lattice constructions.

Wang et al.~\cite{2} incorporated the Bonsai Tree technique to support trapdoor delegation and achieve \textit{NT} and \textit{PSI}, but their construction remains inefficient due to the use of GPV trapdoors. It results in large $\mathcal{O}(n^2)$-sized keys and signatures. In Noh et al.~\cite{14}, although the authors utilized a different trapdoor mechanism (from~\cite{6}) along with a lattice-based chameleon hash function, the use of high-dimensional matrices resulted in a signature and key size of $\mathcal{O}(n^2)$. Cai et al.~\cite{4} proposed an optimization using a filtering technique to reduce the signature size, but the improvement was marginal, and the scheme still results in a quadratic signature size. Similarly, Zhang et al.~\cite{1} introduced the first lattice-based SDVS supporting \textit{ND}, but their use of GPV trapdoors in standard lattices again incurs $\mathcal{O}(n^2)$ complexity in all parameters. Thus, \sdvs is the most efficient SDVS till date in the PQC. \sdvs provides the smallest private key size and signature size among the existing PQC-based SDVS design. \sdvs is compact and well suited for resource constrained devices.

\begin{table}[h]
    \centering
    \caption{Comparison of lattice-based SDVS schemes}
    \begin{tabular}{|c|c|c|c|c|c|c|c|c|}
         \hline
         \textbf{Schemes}   &$\mathbf{|sk_S|}$   &$\mathbf{|sk_V|}$   &$\mathbf{\sf |sig|}$ 
         &\textit{CMA}   &\textit{NT}    &\textit{PSI}   &\textit{ND}   &\textbf{Model}\\
         \hline
         Wang et al. \cite{2}  &$\mathcal{O}(n^2)$    &$\mathcal{O}(n^2)$ &$\mathcal{O}(n^2)$    &EU    &Yes    &Yes   &No    &ROM\\
         \hline
         Noh et al. \cite{14}    &$\mathcal{O}(n^2)$     &$\mathcal{O}(n^2)$     &$\mathcal{O}(n^2)$   &EU  &Yes    &Yes    &No     &SM\\
         \hline
         Cai et al. \cite{4}     &$\mathcal{O}(n^2)$     &$\mathcal{O}(n^2)$     &$\mathcal{O}(n^2)$   &EU   &Yes    &Yes    &No     &ROM\\
         \hline
         Zhang et al. \cite{1}    &$\mathcal{O}(n^2)$     &$\mathcal{O}(n^2)$     &$\mathcal{O}(n^2)$   &SU   &Yes    &Yes    &Yes     &ROM\\
         \hline
         \sdvs    &$\mathcal{O}(\log q)$     &$\mathcal{O}(n\log q)$     &$\mathcal{O}(n\log q)$   &SU   &Yes    &Yes    &Yes     &ROM\\
         \hline
    \end{tabular}
    \label{tab:compare}
\end{table}

\section{Security Analysis of \sdvs}
\begin{theorem}
\sdvs is strongly unforgeable under the adaptive chosen-message attack (\textit{SU-CMA}).
\end{theorem}
\begin{proof}
  We show that \sdvs is \textit{SU-CMA} secure. Refer to Section \ref{sec:su-cma} for the detailed definition of \textit{SU-CMA} game. We assume that there exists a PPT adversary \(\mathcal{A}\) who makes queries to the signer and the random oracle \(\mathcal{H}\) and then performs an adaptive chosen-message attack on our proposed SDVS scheme. The adversary, at last, is able to output a forged designated verifiable signature \({\sf sig}^*\) on a message \(\mu^*\) with a non-negligible probability \(\epsilon\), then we show that there exists a simulator \(\mathcal{B}\) that can exploit \(\mathcal{A}'s\) success probability to solve an instance of the {\sf ring-SIS} problem \((\mathbf{a|b})\cdot \mathbf{v} = 0 \pmod q\). We define a hybrid SDVS, to provide a valid simulation of the random oracle to the adversary \(\mathcal{A}\) in the proof, in which no private keys are used and the output of which can not be distinguished from a real SDVS. In both, signing and simulating algorithms, a hybrid SDVS is generated as follows 
\begin{enumerate}
    \item Sample \(s\xleftarrow{\$} \mathcal{R}_q\) and \(\mathbf{e}\xleftarrow{\$} \mathcal{D}_{{\mathcal{R}_q^{l+k}}, \sigma}\).
    \item Define \(\mathbf{c_0} = \mathbf{b}s + \mathbf{e} \pmod q\).
    \item Sample \(c_1\xleftarrow{\$} \mathcal{R}_q\) and \(\mathbf{z}\xleftarrow{\$} \mathcal{D}_{{\mathcal{R}_q^{l+k}}, \eta\sigma}\).
    \item Program the oracle \(\mathcal{H}(\mathbf{a \cdot z} - tc_1 + \mathbf{b_1 \cdot e} || t || s || \mu) = c_1\).
    \item Output \({\sf sig} = (\mathbf{c_0}, c_1, \mathbf{z})\).
\end{enumerate}

\noindent Throughout the proof, let \(\mathcal{D_{\mathcal{H}}} =\{ c_1 \in \mathcal{R}_q : 0< ||c_1 ||< \sqrt{\kappa}\)\} denotes the range of the random oracle \(\mathcal{H}\) and \(w\) be bound on the number of times the oracle is called or programmed during \(\mathcal{A}'s\) attack i.e. the random oracle query can be made by the adversary \(\mathcal{A}\) directly or the random oracle can be programmed by the signing and simulating algorithms when \(\mathcal{A}\) asks for a signature on some adaptively chosen messages. The interaction between \(\mathcal{A}\) and \(\mathcal{B}\) are as follows:
\begin{itemize}
    \item\textsf{Setup:} Given \(\mathbf{a, b} \in \mathcal{R}_q^{l+k}\), \(\mathcal{B}\) computes the following steps
    \begin{enumerate}
        \item Sample \(\mathbf{s}\xleftarrow{\$}\{-d, \ldots, 0, \ldots, d\}^{l+k}\) and define \(\mathbf{t}= \mathbf{a\cdot s} \pmod q\).
        \item Execute {\sf ringGenTrap} with inputs $\mathbf{\tilde{b}_0}$ and $h_0=1$ to output ${\mathbf{{b}_0}}\in \mathcal{R}_q^{l+k}$ with a trapdoor $\mathbf{R}_{\mathbf{{b}_0}} \in \mathcal{R}_q^{l\times k}$.
        \item Sample random coins \(\phi\) for \(\mathcal{A}\) and \(\Phi\) for \(\mathcal{B}\).
        \item Sample uniformly random \(r_1, r_2, \dots, r_w \xleftarrow{\$} \mathcal{D_H}\), which will correspond to the responses of the random oracle.
        \item Define \((pk_S, sk_S) = (t, \mathbf{s})\) and \((pk_V, sk_V) = (\mathbf{b_0}, \mathbf{b_1}=\mathbf{b}, \mathbf{R_{b_0}}, \emptyset)\).
        \item Keep \((\mathbf{s, R_{b_0}}, r_1, r_2, \dots, r_w, \Phi)\) in secret, and send the public parameters \(pp = (\mathbf{a, b_0, b_1}, t, \phi)\) to \(\mathcal{A}\).
    \end{enumerate}

    \item\textsf{Hash Queries:} Given \((\mathbf{a\cdot y + b_1 \cdot e} ||  t || s || \mu)\), the simulator \(\mathcal{B}\) executes the following steps-
    \begin{enumerate}
        \item Check whether a corresponding hash output \(c\) is stored in the hash list \(l_{\mathcal{H}}\). If yes, return it directly.
        \item Else, sample a random \(r_i \xleftarrow{\$} \{r_1, r_2, \dots, r_w\}\) that has not been used yet. \break Store \((\mathbf{a\cdot y + b_1 \cdot e} || t || s || \mu || c = r_i)\) in \(l_{\mathcal{H}}\) and return it to \(\mathcal{A}\).
    \end{enumerate}

    \item\textsf{Signing Queries:} Given a message $\mu$, \(\mathcal{B}\) executes the following steps
    \begin{enumerate}
        \item Run the signing algorithm in Hybrid SDVS using the random coins \(\Phi\) to produce a signature \({\sf sig} = (\mathbf{c_0}, c_1, \mathbf{z})\).
        \item Store the hash inputs \((\mathbf{a \cdot z} - tc_1 + \mathbf{b_1 \cdot e}|| t || s  || \mu)\) and output \(c_1\) into \(l_{\mathcal{H}}\).
        \item Return \({\sf sig}\) to \(\mathcal{B}\).
    \end{enumerate}
    
    \item\textsf{Simulating Queries:} Given a message \(\mu\), \(\mathcal{B}\) does as in the \(\textsf{Signing Queries}\) and returns a signature \({\sf sig} = (\mathbf{c_0}, c_1, \mathbf{z})\) to \(\mathcal{B}\).

    \item\textsf{Verification Queries:} Given a signature \({\sf sig} = (\mathbf{c_0}, c_1, \mathbf{z})\) and a message \(\mu\), \(\mathcal{B}\) executes the following steps-
    \begin{enumerate}
        \item Check whether \(0< ||\mathbf{z}||< \eta \sigma\sqrt{l+k}\) where \(1<\eta<2\) is a constant.
        \item Run \({\sf ringInvert}\) with the inputs \(\mathbf{c_0, b_0}\), and \(\mathbf{R_{b_0}}\) to output \(s, \mathbf{e}\).
        \item Check whether \((\mathbf{a \cdot z} - tc_1 + \mathbf{b_1 \cdot e} || t || s || \mu|| c_1)\) is stored in the hash list \(l_{\mathcal{H}}\).
        \item Output 1 if all the above are satisfied, otherwise 0.
    \end{enumerate}

    \item\textsf{Output:} With a probability \(\epsilon\), \(\mathcal{A}\) outputs a message-signature pair \((\mu^*, {\sf sig}^* \allowbreak =(\mathbf{c_0}^*, c^*_1, \mathbf{z}^*))\) which satisfies \(0< ||\mathbf{z}^*||< \eta \sigma\sqrt{l+k}\) and \(c^*_1 = \mathcal{H}(\mathbf{a \cdot z^*} - tc^*_1 + \mathbf{b_1 \cdot e^*}|| t || s^* || \mu^*)\) where \(\mathbf{e}^*\in \mathcal{R}^{l+k}_q\) and \(s^*\in \mathcal{R}_q\) had been used in \(\mathbf{b_0}^*\).
\end{itemize}

With a probability \(1-1/|\mathcal{D_{H}}|\), \(c_1^*\) must be one of the \(r_i \in \{r_1, r_2, \dots, r_w\}\). The success probability of \(\mathcal{A}\) in the forgery with the condition that \(c_1^*\) must be one of the \(r_i's\), is at least \(\epsilon-1/|\mathcal{D_{H}}|\). If \(c_1^*= r_i \in \{r_1, r_2, \dots, r_w\}\), then there will be two cases:
\begin{description}
    \item[Case I] If \(c_1^*= r_i\) was programmed during the \textsf{Signing Queries} or the \textsf{Simulating Queries}, then we analyse this as follows. Assume that the simulator \(\mathcal{B}\) programmed the random oracle \(c^*_1 = \mathcal{H}(\mathbf{a \cdot z} - tc^*_1 + \mathbf{b_1 \cdot e}|| t|| s || \mu)\) when signing a message \(\mu\). Since the adversary outputs a valid forgery \(({\sf sig}^*=(\mathbf{c_0}^*, c^*_1, \mathbf{z}^*))\) for \(\mu^*\), we have \(\mathcal{H}(\mathbf{a \cdot z} - tc^*_1 + \mathbf{b_1 \cdot e}|| t|| s || \mu) = \mathcal{H}(\mathbf{a \cdot z^*} - tc^*_1 + \mathbf{b_1 \cdot e^*}|| t|| s^* || \mu^*)\). Thus according to the collision-resistance property of \(\mathcal{H}\), \(\mathcal{B}\) computes \(\mathbf{e}, \mathbf{e}^*\) from the values \(\mathbf{c_0}, \mathbf{c_0}^*\) using the subroutine 
\({\sf ringInvert}\) respectively.So, \(\mathcal{B}\) computes 
\[\mathbf{a \cdot z} - tc^*_1 + \mathbf{b_1 \cdot e} = \mathbf{a \cdot z} - tc^*_1 + \mathbf{b \cdot e} = \mathbf{a \cdot z^*} - tc^*_1 + \mathbf{b_1 \cdot e^*}\]
\[\implies \mathbf{a(z-z^*+b(e-e^*))= 0}\pmod q\]
\[\implies \mathbf{(a|b)\begin{bmatrix}
\mathbf{z-z^*} \\ 
\mathbf{e-e^*} 
\end{bmatrix}= 0} \pmod q\]

Let \(\mathbf{v = \begin{bmatrix}
\mathbf{z-z^*} \\ 
\mathbf{e-e^*} 
\end{bmatrix}}\), it can be easily verified that \(\mathbf{v\neq 0}\), because if it is 0 then the forged signature \({\sf sig^*}\) is exactly same as the old signature \({\sf sig}\). Also, since \(0< ||\mathbf{z}^*||, ||\mathbf{z}||< \eta \sigma\sqrt{l+k}\) and \(0< ||\mathbf{e}^*||, ||\mathbf{e}||< \sigma\sqrt{l+k}\), we have that \(0< ||\mathbf{v}||< 2\sigma\sqrt{(\eta^2+1)l+k}=  \mathcal{O}(n)\).

    \item[Case II] If \(c_1^* = r_i\) was a response to the \textsf{Hash Queries}, then we analyse this as follows. The simulator \(\mathcal{B}\) records the forged signature \({\sf sig}^*= (\mathbf{c_0}^*, c_1^*, \mathbf{z}^*)\) given by the adversary \(\mathcal{A}\) on the message \(\mu^*\) and generates new random values \(r'_i, r'_{i+1}, \ldots, r'_w \allowbreak \xleftarrow{\$} \mathcal{D_H}\). The simulator \(\mathcal{B}\) then runs again with inputs \((\mathbf{a, b_0, b_1, R_{b_0}}, t, r_1, \ldots, r_{i-1},\allowbreak  r'_i, r'_{i+1}, \ldots, r'_w, \phi, \Phi)\). By the General Forking Lemma, \(\mathcal{B}\) observes the probability that \(r'_i\neq r_i\)  and \(\mathcal{A}\) uses the random oracle response in the forgery is at least \(\epsilon^*\) and hence with the probability \(\epsilon^*\), \(\mathcal{A}\) outputs a new signature \({\sf sig} =  (\mathbf{b_0}, r'_i, \mathbf{z})\) of the message \(\mu^*\) and \((\mathbf{a\cdot z^*} - tc_1^* + \mathbf{b_1\cdot e^*}|| t|| s^*) = (\mathbf{a\cdot z} - tc_1 + \mathbf{b_1\cdot e}|| t|| s)\) where \(c_1^* = r_i, c_1= r'_i\) and \(s=s^*\). \(\mathcal{B}\) also computes \(\mathbf{e}\) and \(\mathbf{e^*}\) from the values \(\mathbf{c_0}\) and \(\mathbf{c^*_0}\). Thus by putting the values \(t = \mathbf{a\cdot s} \pmod q\), we have 

\[(\mathbf{a\cdot z}- t c_1 +\mathbf{b_1\cdot e} )= (\mathbf{a\cdot z}- t c_1 +\mathbf{b \cdot e}) = (\mathbf{a\cdot z^*}- t c^*_1 +\mathbf{b \cdot e^*})\]

\[\implies \mathbf{a\cdot (z-z^*+ s}c_1^*-\mathbf{s}c_1)+\mathbf{b \cdot (e - e^*) = 0} \pmod q\]

\[\implies \mathbf{(a|b)\begin{bmatrix}
\mathbf{z-z^* + s}(c_1^* - c_1) \\ 
\mathbf{e-e^*} 
\end{bmatrix}= 0} \pmod q\]

Let \(\mathbf{v = }\begin{bmatrix}
\mathbf{z-z^* + s}(c_1^* - c_1) \\ 
\mathbf{e-e^*} 
\end{bmatrix}\), since \(||\mathbf{e^*}||, ||\mathbf{e}|| \leq \sigma \sqrt{l+k};\;||\mathbf{z^*}||, ||\mathbf{z}||\leq \eta \sigma \allowbreak \sqrt{l+k};\; ||\mathbf{s}c_1^*||, ||\mathbf{s}c_1||\leq d\kappa \sqrt{l+k}\), we have that \(||\mathbf{v||}\leq 2\sqrt{l+k}\sqrt{(\eta\sigma+d\kappa)^2 + \sigma^2} = \mathcal{O}(n)\).  

\end{description} 
\end{proof}

\begin{theorem}
    \sdvs is Non-Transferable
\end{theorem}
\begin{proof}
    According to the defined security model for non-transferability of the SDVS scheme, a game between the PPT adversary $\mathcal{A}$ and the simulator \(\mathcal{B}\) is defined as follows
 \begin{itemize}
     \item\textsf{Setup:} \(\mathcal{B}\) computes the following steps.
     \begin{enumerate}
         \item Run \textsf{Setup(\(1^n\))} to output the public parameter \(pp\).
         \item Run \textsf{SignKeyGen(\(pp\))} to output \((pk_S, sk_S) = (t, \textsf{s})\) for the signer.
         \item Run \textsf{VerKeyGen\((pp)\)} to output $(pk_V, sk_V) = \mathbf{(b_0, b_1, R_{b_0}, R_{b_1})}$.
         \item Send \((pk_S, pk_V)\) to \(\mathcal{A}\).

     \end{enumerate}

    \item\textsf{Challenge:} For a given message \(\mu\) from the adversary \(\mathcal{A}\), \(\mathcal{B}\) computes the following steps.
    \begin{enumerate}
        \item Run \textsf{Sign\((pp,sk_S, pk_S, pk_V, \mu)\)} to output a real designated verifiable signature \({\sf sig}^{(0)} = (\mathbf{c_0}^{(0)}, c^{(0)}_1,  \mathbf{z}^{(0)})\).
        \item Run \textsf{Simul\((pp, pk_S, sk_V, pk_V, \mu)\)} to output a simulated designated verifiable signature \({\sf sig}^{(1)} = (\mathbf{c_0}^{(1)}, c^{(1)}_1,  \mathbf{z}^{(1)})\).
        \item Sample a random bit \(b\xleftarrow{\$}\{0, 1\}\) and send the message-signature pair \((\mu, {\sf sig}^{(b)})\) to \(\mathcal{A}\).

    \end{enumerate}

    \item\textsf{Output:} \(\mathcal{A}\) outputs a bit \(b^*\in \{0, 1\}\).
 \end{itemize}
In the above processes, by the property of hash function \(\mathcal{H}\), we assure that $c^{(0)}_1$ and $c^{(1)}_1$ are random in \(\mathcal{R}_q\). As the output of the real designated verifiable signature is done using the rejection sampling, the pairs \((c^{(0)}_1, \;\mathbf{z^{(0)}= s\cdot} c^{(0)}_1 + \mathbf{y})\) and \((c^{(1)}_1, \; \mathbf{z^{(1)}\xleftarrow{\$} \mathcal{D}_{\mathcal{R}^{l+k}_q, \sigma}})\) are within the statistical distance of \(2^{-\omega(log(l+k))}/M\). Furthermore, \(\mathbf{c^{(0)}_0 = b_0\cdot}s^{(0)} + \mathbf{e}^{(0)}\) and \(\mathbf{c^{(1)}_0 = b_0\cdot}s^{(1)} + \mathbf{e}^{(1)}\) where \(s^{(0)}, s^{(1)}\xleftarrow{\$} \mathcal{R}_q,\; \mathbf{e}^{(0)}\xleftarrow{\$} \mathcal{D}_{\mathcal{R}^{l+k}_q, \sigma}\), and \(\mathbf{e}^{(1)}\in \mathcal{D}_{\mathcal{R}^{l+k}_q, \sigma}\) is coming as an output from the {\sf RingSample} algorithm, hence \(\mathbf{c^{(0)}_0}\) and \(\mathbf{c^{(1)}_0}\) are statistically indistinguishable. 
The real \((\mu, {\sf sig}^{(0)})\) and the simulated \((\mu, {\sf sig}^{(1)})\) designated verifiable signatures are statistically indistinguishable to \(\mathcal{A}
\) and hence \(\mathcal{A}
\) can not determine who is the actual producer of the pair \((\mu, {\sf sig}^{(b)})\) which implies that the advantage of the adversary $\mathcal{A}$ is negligible.
\end{proof}

\begin{theorem}
\sdvs satisfies the \textit{PSI} security if the {\sf Ring-LWE} assumption holds.
\end{theorem}
\begin{proof}
    According to the defined security model for \textit{PSI} of the SDVS scheme, a game between the PPT adversary $\mathcal{A}$ and the simulator \(\mathcal{B}\) is defined as follows
\begin{itemize}
    \item\textsf{Setup:} \(\mathcal{B}\) computes the following steps-
    \begin{enumerate}
        \item Run \textsf{Setup(\(1^n\))} to output the public parameter \(pp\).
         \item Run \textsf{SignKeyGen(\(pp\))} twice to output \((pk_{S^{(0)}}, sk_{S^{(b)}}) = (t_0, \textsf{s}_0)\) for the signer \(\mathbf{S}_0\) and \break \((pk_{S^{(1)}}, sk_{S^{(1)}}) = (t_1, \textsf{s}_1)\) for the signer \(\mathbf{S}_1\)
         \item Run \textsf{VerKeyGen\((pp)\)} to output $(pk_V, sk_V) = \mathbf{(b_0, b_1, R_{b_0}, R_{b_1})}$.
         \item Send \((pp, pk_{S^{(0)}}, sk_{S^{(0)}}, pk_{S^{(1)}}, sk_{S^{(1)}}, pk_V)\) to \(\mathcal{A}\).
    \end{enumerate}
    \item\textsf{Simulating Queries:} For a message \(\mu\) and a bit \(b\in\{0,1\}\) given by the adversary \(\mathcal{A}\), \(\mathcal{B}\) runs \(\mathsf{Simul}(pp, sk_V, pk_V, pk_{S^{(b)}}, \mu)\) and returns a simulated designated verifiable signature \({\sf sig}^{(b)} = (\mathbf{c_0}^{(b)}, c_1^{(b)}, \mathbf{z}^{(b)})\) to \(\mathcal{A}\).

    \item\textsf{Verification Queries:} Given a message-signature pair \((\mu, {\sf sig}^{(b)})\) and a bit \(b\in \{0,1\}\) from the adversary, \(\mathcal{B}\) runs the \(\mathsf{Verify}(pp, sk_V, pk_V, pk_{S^{(b)}, {\sf sig}^{(b)}}, \mu)\) and returns a decisional value 1 for valid and 0 otherwise to \(\mathcal{A}\). 

    \item\textsf{Challenge:} For a message \(\mu\), \(\mathcal{B}\) sends the challenge as follows-
    \begin{enumerate}
        \item Run \(\mathsf{Sign}(pp, pk_{S^{(0)}}, sk_{S^{(0)}}, pk_V, \mu)\) to output \({\sf sig}^{(0)}=(\mathbf{c_0}^{(0)}, c^{(0)}_1, \mathbf{z}^{(0)})\).
        \item Run \(\mathsf{Sign}(pp, pk_{S^{(1)}}, sk_{S^{(1)}}, pk_V, \mu)\) to output \({\sf sig}^{(1)}=(\mathbf{c_0}^{(1)}, c^{(1)}_1, \mathbf{z}^{(1)})\).
        \item Sample a random bit \(b\xleftarrow{\$}\{0, 1\}\) and send the pair \((\mu, {\sf sig}^{(b)})\) to \(\mathcal{A}\).
    \end{enumerate}

    \item\textsf{Output:} \(\mathcal{A}\) outputs a bit \(b^*\in\{0, 1\}\).
\end{itemize}

\noindent Similar to the arguments in the proof of non-transferability, by the property of hash function \(\mathcal{H}\), we assure that $c^{(0)}_1$ and $c^{(1)}_1$ are random in \(\mathcal{R}_q\). As the output of the real designated verifiable signature is done using the rejection sampling, the pairs \((c^{(0)}_1, \mathbf{z^{(0)}= \mathbf{s}^{(0)}\cdot} c^{(0)}_1 + \mathbf{y}^{(0)} \pmod q)\) and \((c^{(1)}_1, \mathbf{z^{(1)}= \mathbf{s}^{(1)}\cdot} c^{(1)}_1 + \mathbf{y}^{(1)})\pmod q)\) are within the same distribution. Furthermore, \(\mathbf{c^{(0)}_0 = b^{(0)}_0\cdot}s^{(0)} + \mathbf{e}^{(0)} \pmod q\) and \(\mathbf{c^{(1)}_0 = b^{(1)}_0\cdot}s^{(1)} + \mathbf{e}^{(1)}\pmod q\) are the LWE instances where \(s^{(0)}, s^{(1)}\xleftarrow{\$} \mathcal{R}_q,\; \mathbf{{e}^{(0)}, {e}^{(1)}}\xleftarrow{\$} \mathcal{D}_{\mathcal{R}^{l+k}_q, \sigma}\), hence \(\mathbf{c^{(0)}_0}\) and \(\mathbf{c^{(1)}_0}\) are statistically close to uniform vectors in \(\mathcal{R}^{l+k}_q\). Now, given a real message-signature pair \((\mu, {\sf sig^{(b)}})\), \(\mathcal{A}\) computes the following steps
\begin{enumerate}
    \item Parse the signature \((\mu, {\sf sig^{(b)}}) = (\mu, \mathbf{c_0}^{(b)}, c^{(b)}_1, \mathbf{z}^{(b)})\).
    \item Check if $0< ||\mathbf{z}^{(b)}||< \eta \sigma\sqrt{l+k}$; where \(\eta\) is a constant.
\end{enumerate}
\noindent According to the definition of the {\sf Ring-LWE} problem, this is obvious that without the private key \(sk_V = (\mathbf{R_{b_0}, R_{b_1}})\) of a designated verifier, the quantities \(s\) and \(\mathbf{e}\) can not be obtained and the verification can not be accomplished by the adversary \(\mathcal{A}\). Hence the PPT adversary \(\mathcal{A}\) can not verify the validity of a given message-signature pair \(\mu, {\sf sig^{(b)}}\) and determine which one is the real signer and thus the advantage of \(\mathcal{A}\) is negligible.
\end{proof}

\begin{theorem}
    \sdvs satisfies Non-Delgatability.
\end{theorem}
\begin{proof}
     As per the security definition defined above, let us consider that there exists a black box \(\mathcal{A}\) producing a designated verifiable signature in time \(\tau'\) with a non-negligible probability \(\epsilon'\), a game between an extractor \(\mathcal{E}\) and \(\mathcal{A}\) is as follows:

-if \(\mathcal{A}\) produces a public key \(pk_S = t \in \mathcal{R}_q\) as the target:
\begin{itemize}
    \item\textsf{Setup:} \(\mathcal{E}\) and \(\mathcal{A}\) perform the following steps:
    \begin{enumerate}
        \item \(\mathcal{E}\) runs \({\sf ringGenTrap}(\mathbf{\tilde{a_0}}, h_{\mathbf{a}} = 1)\) to output \(\mathbf{a}\in \mathcal{R}_q^{l+k}\) and a trapdoor \(\mathbf{R_{a}} \in \mathcal{R}_q^{l\times k}\), and sends \(\mathbf{a}\) to $\mathcal{A}$.
        
        \item \(\mathcal{A}\) produces \(pk_S = t\) and sends it to \(\mathcal{E}\).
        
        \item \(\mathcal{E}\) runs \({\sf ringGenTrap}(\mathbf{\tilde{b_0}}, h_{\mathbf{b_0}} = 1)\) to output \(\mathbf{b_0}\in \mathcal{R}_q^{l+k}\) and a trapdoor \(\mathbf{R_{b_0}} \in \mathcal{R}_q^{l\times k}\).

        \item \(\mathcal{E}\) runs \({\sf ringGenTrap}(\mathbf{\tilde{b_1}}, h_{\mathbf{b_1}} = 1)\) to output \(\mathbf{b_1}\in \mathcal{R}_q^{l+k}\) and a trapdoor \(\mathbf{R_{b_1}} \in \mathcal{R}_q^{l\times k}\).

        \item \(\mathcal{E}\) samples random coins \(\Phi\), and uniformly random \(r_1, \dots, r_w \xleftarrow{\$} \mathcal{D_H}\), which will correspond to the responses of the random oracle.

        \item \(\mathcal{E}\) sends the \(pk_V = (\mathbf{b_0, b_1})\) to \(\mathcal{A}\).
    \end{enumerate}

    \item\textsf{Hash Queries:} \(\mathcal{E}\) performs this in the same way as the simulator \(\mathcal{B}\) does in the \textit{SU-CMA} proof.

    \item\textsf{Simulating Queries:} Given a message \(\mu\), \(\mathcal{E}\) specifies the following steps:
    \begin{enumerate}
        \item Sample \(s, u \xleftarrow{\$} \mathcal{R}_q,\;\) and \(\mathbf{z} \xleftarrow{\$} \mathcal{D}_{\mathcal{R}^{l+k}_q, \eta \sigma}\).

        \item Program the random oracle \(\mathcal{H}(u|| t|| s|| \mu) = c_1\).

        \item Run \({\sf ringSample}(\mathbf{b_1}, h_{\mathbf{b_1}}, \mathbf{R_{b_1}}, u-\mathbf{a\cdot z} + tc_1, \sigma)\) to get a short \(\mathbf{e} \in \mathcal{D}_{\mathcal{R}^{l+k}_q, \sigma}\).
        
        \item Define \(\mathbf{c_0 = b_0}s + \mathbf{e} \pmod q\).

        \item Output \({\sf sig} = (\mathbf{c_0}, c_1, \mathbf{z})\) using rejection sampling technique.
        
    \end{enumerate}

    \item\textsf{Verification Queries:} Similar to the \textsf{Hash Queries}, extractor \(\mathcal{E}\) performs this in the same way as the simulator \(\mathcal{B}\) does in the \textit{SU-CMA} proof.

    \item\textsf{Challenge:} \(\mathcal{E}\) selects \(\mu^*\) and \(\mathcal{A}\) sends the SDVS \({\sf sig}^* = (\mathbf{c^*_0}, c^*_1, \mathbf{z^*})\) for \(\mu^*\) to \(\mathcal{E}\).

    \item\textsf{Output:} Given \((\mu^*, {\sf sig}^*)\), \(\mathcal{E}\) computes-
    \begin{enumerate}
        \item Check the validity of \({\sf sig}^*\) on \(\mu^*\).
        \item Run \({\sf ringSample} (\mathbf{a, R_a}, t, d)\) to get a short \(\mathbf{s} \in \mathcal{D}_{\mathcal{R}_q^{l+k},d}\).
    \end{enumerate}
\end{itemize}

Thus, the extractor \(\mathcal{E}\) determines the private key \(sk_S = \mathbf{s}\) of \(\mathcal{A}\) which is the pre-image of the public key \(pk_S = t\). The probability of success is the same as the probability of \(\mathcal{A}\) producing a valid signature in the challenge phase.

-if \(\mathcal{A}\) produces the public key \(pk_V = (\mathbf{b_0, b_1})\) as the target:
\begin{itemize}
    \item\textsf{Setup:} \(\mathcal{E}\) and \(\mathcal{A}\) do the following steps:
    \begin{enumerate}
        \item \(\mathcal{E}\) samples \(\mathbf{a} \xleftarrow{\$} \mathcal{R}_q^{l+k}\) and sends it to \(\mathcal{A}\).

        \item \(\mathcal{A}\) sends \(pk_V = (\mathbf{b_0, b_1})\) to \(\mathcal{E}\).

        \item \(\mathcal{E}\) samples a uniformly random \(\mathbf{s} \xleftarrow{\$} \{-d, \ldots, 0, \ldots, d\}^{l+k}\).

        \item \(\mathcal{E}\) samples random coin \(\Phi\), and uniformly random \(\{r_1, \ldots, r_w\} \xleftarrow{\$} \mathcal{D_H}\), that will correspond to the output of the random oracle.

        \item \(\mathcal{E}\) defines \(t = \mathbf{a\cdot s} \pmod q\), and sends \(pk_S = t\) to \(\mathcal{A}\).
    \end{enumerate}

    \item\textsf{Hash Queries:} \(\mathcal{E}\) performs this in the same way as the simulator \(\mathcal{B}\) does in the \textit{SU-CMA} proof.

    \item\textsf{Signing Queries:} Given a message \(\mu\), \(\mathcal{E}\) executes the following steps-
    \begin{enumerate}
        \item Run the signing algorithm in Hybrid SDVS, as defined in the proof of \textit{SU-CMA}, using the random coins \(\Phi\) to produce a signature \({\sf sig} = (\mathbf{c_0}, c_1, \mathbf{z})\).
        \item Store the hash inputs \((\mathbf{a \cdot z} - tc_1 + \mathbf{b_1 \cdot e}|| t|| s || \mu)\) and output \(c_1\) into \(l_{\mathcal{H}}\).
        \item Return \({\sf sig}\) to \(\mathcal{A}\).
    \end{enumerate}

    \item\textsf{Challenge:} \(\mathcal{E}\) selects \(\mu^*\), and \(\mathcal{A}\)  sends the corresponding SDVS \({\sf sig}^* = (\mathbf{c^*_0}, c^*_1, \mathbf{z^*})\) to \(\mathcal{E}\).

    \item\textsf{Output:} The extractor \(\mathcal{E}\) first checks the validity of \({\sf sig}^* = (\mathbf{c^*_0}, c^*_1, \mathbf{z^*})\) on \(\mu^*\) and adopts the same method as the simulator \(\mathcal{B}\) does in the proof of \textit{SU-CMA} by programming the random oracle \(\mathcal{H}\) while signing a message. So, we have, due to the collision-resistant property of the hash function \(\mathcal{H}\),
    
    \[(\mathbf{a\cdot z^*} - tc_1^* + \mathbf{b_1 \cdot e^*}|| t|| s^*|| \mu^*) = (\mathbf{a\cdot z} - tc_1 + \mathbf{b_1 \cdot e}|| t|| s|| \mu)\] 
    by holding the conditions \(c_1^* = c_1, \mathbf{z^* = z},\) and \(s^* = s\). The extractor \(\mathcal{E}\), using the values \(\mathbf{b_0}, s\) and \(\mathbf{ b_0^*}, s^*\) computes the values \(\mathbf{e}\) and \(\mathbf{e^*}\) respectively. Hence, we have 
    \[\mathbf{a\cdot z} - tc_1 + \mathbf{b_1 \cdot e} = \mathbf{a\cdot z^*} - tc_1^* + \mathbf{b_1 \cdot e^*}\]

    \[\implies \mathbf{b_1 \cdot (e^* - e)} = \mathbf{a \cdot (z-z^*)} + t(c_1^* - c_1) \pmod q\] 
    
    \[\implies \mathbf{b_1 \cdot (e^* - e)} = 0 \pmod q\].

    Now, let us assume that \(\mathbf{v = e^* - e}\), then it can be easily checked that \(\mathbf{v} \neq 0\), otherwise, if it is not so, then the forged signature \({\sf sig^*}\) is exactly as same as the programmed signature \({\sf sig}\) provided by the challenger during the signing queries. Therefore, the term \(\mathbf{v = (e^* - e)} \in \mathcal{R}_q^{l+k}\) is a vector of polynomials of short norm. After replaying the above process at least $k$ times, the extractor can get a matrix $\mathbf{V} \in \mathcal{R}_q^{(l+k)\times k}$, columns of which are of short norms.

    Since, $\mathbf{b_1} = (\mathbf{b}^T_{1_0}, \mathbf{b}^T_{1_1} = h\mathbf{g}^T - \mathbf{b}^T_{1_0}\mathbf{R_{b_1}})^T \in \mathcal{R}_q^{l+k}$,\; where $\mathbf{b}^T_{1_0} \in \mathcal{R}_q^l$ and  $\mathbf{b}^T_{1_1} \in \mathcal{R}_q^k$.  Therefore $\mathbf{b_1}\cdot \mathbf{V} = (\mathbf{b}^T_{1_0}, \mathbf{b}^T_{1_1} = h\mathbf{g}^T - \mathbf{b}^T_{1_0}\mathbf{R_{b_1}})^T \cdot \mathbf{V} = 0 \pmod q$ 

    Let $\mathbf{V} = \begin{bmatrix}
        \mathbf{V^{(0)}} \\ 
        \mathbf{V^{(1)}} 
    \end{bmatrix}$, where $\mathbf{V^{(0)}} \in \mathcal{R}_q^{l\times k}$ and $\mathbf{V^{(1)}} \in \mathcal{R}_q^{k\times k}$.

    So, solving the equation 
    \[\mathbf{b_1}\cdot \mathbf{V} = (\mathbf{b}^T_{1_0}, \mathbf{b}^T_{1_1} = h\mathbf{g}^T - \mathbf{b}^T_{1_0}\mathbf{R_{b_1}})^T \cdot \begin{bmatrix}
        \mathbf{V^{(0)}} \\ 
        \mathbf{V^{(1)}} 
    \end{bmatrix} = 0 \pmod q\]

    \[\implies \mathbf{b}^T_{1_0} \cdot \mathbf{V^{(0)}} + (h\mathbf{g}^T - \mathbf{b}^T_{1_0}\mathbf{R_{b_1}})^T \cdot \mathbf{V^{(1)}} = 0 \pmod q\]
   the extractor $\mathcal{E}$ can retrieve the private key $\mathbf{R_{b_1}}$ of the adversary $\mathcal{A}$.

\end{itemize}
\end{proof}

\section{Conclusion}
In this work, we propose a post-quantum SDVS denoted as \sdvs based on ideal lattice assumptions, namely {\sf Ring-SIS} and {\sf Ring-LWE}. \sdvs provides strong security guarantees, including strong unforgeability under chosen-message attacks, non-transferability, non-delegatability, and signer anonymity. The signature and key sizes were minimized without compromising security. A signature size of $\mathcal{O}(n \log q)$ was achieved, which constituted a quadratic reduction compared to the conventional $\mathcal{O}(n^2)$ lattice-based SDVS schemes, resulting in a reduction by a factor of $n / \log q$. It is demonstrated that \sdvs outperforms existing post-quantum SDVS designs in terms of efficiency and compactness.

\bibliographystyle{plain}
\bibliography{ref.bib}

\end{document}